\documentclass[reqno]{amsart}
\usepackage{amssymb,graphics,graphicx,bbm,hyperref, color,comment, dsfont}

\def\be{\begin{equation}}
\def\ee{\end{equation}}
\def\bm{\begin{multline}}
\def\bfig{\begin{figure}[htb]}
\def\efig{\end{figure}}
\newcommand{\dd}{{\rm d}}
\newcommand{\e}[1]{\,{\rm e}^{#1}\,}
\newcommand{\ii}{{\rm i}}

\def\Tr{{\operatorname{Tr\,}}}

\numberwithin{equation}{section}
\newtheorem{theorem}{Theorem}[section]

\newcommand{\caB}{{\mathcal B}}
\newcommand{\caC}{{\mathcal C}}
\newcommand{\caE}{{\mathcal E}}

\newcommand{\caH}{{\mathcal H}}
\newcommand{\caL}{{\mathcal L}}

\newcommand{\caS}{{\mathcal S}}
\newcommand{\bbC}{{\mathbb C}}

\newcommand{\bbN}{{\mathbb N}}
\newcommand{\bbP}{{\mathbb P}}
\newcommand{\bbR}{{\mathbb R}}

\newcommand{\bbZ}{{\mathbb Z}}

\makeatletter
  \def\tagform@#1{\maketag@@@{\footnotesize{(#1)}\@@italiccorr}}
\makeatother
\renewcommand{\eqref}[1]{(\ref{#1})}

\begin{document}


\title[Decay of correlations in 2D quantum systems]{Decay of correlations in 2D quantum systems with continuous symmetry}

\author{Costanza Benassi}
\address{Department of Mathematics, University of Warwick,
Coventry, CV4 7AL, United Kingdom}
\email{c.benassi@warwick.ac.uk}

\author{J\"urg Fr\"ohlich}
\address{Institut f\"ur Theoretische Physik, ETH Z\"urich, Wolfgang-Pauli-Str.\ 27,
8093 Z\"urich, Switzerland}
\email{juerg@phys.ethz.ch}

\author{Daniel Ueltschi}
\address{Department of Mathematics, University of Warwick,
Coventry, CV4 7AL, United Kingdom}
\email{daniel@ueltschi.org}



\begin{abstract}
We study a large class of models of two-dimensional quantum lattice systems with continuous symmetries, and we prove a general McBryan-Spencer-Koma-Tasaki theorem concerning algebraic decay of correlations. We present applications of our main result to the Heisenberg-, Hubbard-, and $t$-$J$ models, and to certain models of random loops.
\end{abstract}

\thanks{\copyright{} 2016 by the authors. This paper may be reproduced, in its
entirety, for non-commercial purposes.}

\maketitle

{\small\tableofcontents}

\section{Introduction}

Absence of spontaneous magnetization in the two-dimensional quantum Hei\-sen\-berg model was proven by Mermin and Wagner in their seminal article \cite{MW}. Their result revealed some fundamental properties of two-dimensional systems with continuous symmetry. It was subsequently extended and generalized in several directions. In  \cite{FP1,FP2}, absence of continuous symmetry breaking in extremal Gibbs states was established for a large class of models. Fisher and Jasnow \cite{FJ} first explained why there isn't any long-range order in some two-dimensional models with continuous symmetries; the two-point correlation functions of the Heisenberg model were shown to decay at least logarithmically. Their decay is, however, expected to be power-law. This was first proven by McBryan and Spencer \cite{MS} for the classical rotor model in a short and lucid article. Their proof is based on complex rotations and extends to a large family of classical spin systems. Shlosman obtained similar results with a different method \cite{Shl}. Power-law decay was established for a wide class of classical systems and some quantum models (such as the ferromagnetic Heisenberg) in \cite{BFK,Ito}; the proofs in these papers involve the Fourier transform of correlations and the Bogolubov inequality, and they are limited to regular two-dimensional lattices. A more general result was obtained by Koma and Tasaki in their study of the Hubbard model, using complex rotations \cite{KT}. Their proof was simplified and applied to the XXZ spin-$\frac{1}{2}$ model on generic two-dimensional lattices in \cite{FU}.
In the work presented in \cite{BFK,FJ,FU,Ito,KT} specific models are considered, and proofs rely on explicit settings. The method of proof in \cite{KT} is however robust, and it can be expected to apply to a much broader class of models. Our goal in the present article is to propose a general setting accommodating many models of interest and prove a general result concerning the decay of correlations. As a consequence, we obtain new results for generalized SU(2)-invariant models with higher spins, for the $t$-$J$ model, and for random loop models, as well as obtaining a bound similar to \cite{KT} for the Hubbard model. It is worth noticing that in our setting the lattice is not necessarily regular --- indeed, our results hold for any two-dimensional graph.

Before describing the general setting, we introduce several explicit models; we start with SU(2)-invariant models of quantum spin systems (Section \ref{sec spins}), then consider some random loop models (Section \ref{sec loops}) and the Hubbard model (Section \ref{sec Hubbard}), and end with the $t$-$J$ model (Section \ref{sec t-J}).
In Section \ref{sec general} the general setting is introduced, and a general theorem is stated and proven.
Applications of our general results to the specific models introduced in Sections \ref{sec spins} through \ref{sec t-J} are presented in Section \ref{sec proofs}.

An open problem is to find a proof for bosonic systems, such as the Bose-Hubbard model. The present method relies on local operators to be bounded, and it does not generalize to bosonic systems in a straightforward way.

\section{Quantum spin systems}
\label{sec spins}

Let $\Lambda$ be a finite graph, with a set of edges denoted by $\mathcal{E}$. One may think of $\Lambda$ as a ``lattice''. The ``graph distance'' is the length of the shortest connected path between two vertices in 
$\Lambda$ and is denoted by $d:\Lambda\times\Lambda\rightarrow\bbN_0$. We consider graphs of arbitrary size, but with a bounded ``perimeter constant'' $\gamma$:
\be
\label{def gamma}
\gamma = \max_{x\in\Lambda}\max_{\ell\in\bbN}\frac{1}{\ell} \bigl| \bigl\{ y\in\Lambda \, |\, d(x,y)= \ell\bigr\}\bigr|,
\ee
which expresses their two-dimensional nature.
Typical examples of allowed graphs are finite subsets of $\bbZ^{2}$, where edges connect nearest-neighbor sites, in which case $\gamma=4$. Further examples are furnished by finite subsets of the triangular, hexagonal, or Kagom\'e lattices. It is worth mentioning that we do not assume translation invariance.

Let $s\in\frac{1}{2}\bbN$, and let $\vec{\mathcal{S}} = \left(\mathcal{S}^1,\mathcal{S}^2, \mathcal{S}^3\right)$ be the vector of spin-$s$ matrices acting on the Hilbert Space $\bbC^{2s+1}$, and satisfying 
$\left[\mathcal{S}^1, \mathcal{S}^2\right] = \ii \mathcal{S}^3$, together with all cyclic permutations, and 
$\sum_{i=1}^{3} (\caS^{i})^{2} = s (s+1)$. Moreover, we define ladder operators $\caS^\pm = \caS^1\pm i \caS^2$.

The most general SU(2) invariant hamiltonian with spin $s$ and pair interactions is of the form
\begin{equation}
\label{HSM}
H_{\Lambda} = -\sum_{\langle x, y\rangle \in \mathcal{E}} \sum_{k=1}^{2s} c_{k}(x,y)\left(\vec{\mathcal{S}}_{x}\cdot\vec{\mathcal{S}}_{y}\right)^k.
\end{equation}
Here the letters $c_k(x,y)$ denote the coupling constants, and $\mathcal{S}^i_x := \mathcal{S}^i \otimes \mathds{1}_{\Lambda\backslash x}$. The hamiltonian $H_{\Lambda}$ acts on the Hilbert space $\mathcal{H}_\Lambda := \bigotimes_{x\in\Lambda}\bbC^{2s+1}$. The Gibbs state at inverse temperature $\beta$ is given by $$\langle (\cdot) \rangle = \Tr \big(( \cdot) \e{-\beta H_{\Lambda}}\big) / \Tr\big(\e{-\beta H_{\Lambda}}\big).$$ 
Without loss of generality we assume that
\be
\sum_{k}|c_k(x,y)| \left(3s^2\right)^k \leq 1
\label{assumptionHSM}
\ee
for all $x,y \in \Lambda$. Invariance under SU(2) implies that the hamiltonian of this model commutes with the component of the total spin operator along any of the three coordinate axes. We have that
\be
\label{comm spin}
\left[ \vec{\mathcal{S}}_{x}\cdot\vec{\mathcal{S}}_{y}, \, \mathcal{S}^i_x+\mathcal{S}^i_y\right]=0,
\ee
for $i=1,2,3$. Actually, we will only exploit invariance of the hamiltonian under rotations around a single axis to get an inverse-power-law bound on the decay of correlations.

\begin{theorem}
\label{thm HSM}
Let $H_{\Lambda}$ be the hamiltonian defined in \eqref{HSM}. There exist constants $C>0$ and 
$\xi(\beta)>0$, the latter depending on $\beta, \gamma, s$ but \textit{not} on $x,y\in\Lambda$, such that
\[
|\langle \mathcal{S}^j_x \mathcal{S}^j_y\rangle| \leq C \left(d(x,y)+1\right)^{-\xi(\beta)}.
\]
More generally, for $\mathcal{O}_y\in\mathcal{B}_y$,
\[
|\langle \mathcal{S}^+_x \mathcal{O}_y\rangle| \leq C \left(d(x,y)+1\right)^{-\xi(\beta)}.
\]
The exponent $\xi(\beta)$ is proportional to $\beta^{-1}$; more precisely,
\[
\lim_{\beta \rightarrow \infty}\beta\,\xi(\beta) = (32s\gamma^2)^{-1}.
\]
\end{theorem}

We could also consider models with interactions that are asymmetric with respect to different spin directions. If such models retain a U(1)-symmetry, our main theorem and its proof can easily be seen to remain valid. In the absence of a non-abelian continuous symmetry we predict the expected behavior for $\xi(\beta)$. Indeed, a Berezinski-Kosterlitz-Thouless transition is expected to take place, the decay of correlations changing from exponential to power law with an exponent proportional to $\beta^{-1}$, for large $\beta$. This has been proven for the classical XY model in \cite{FS}. For models with SU(2)-symmetry, one expects exponential decay for all positive temperatures.

The proof of Theorem \ref{thm HSM} can be found in Section \ref{sec proofs}.

\section{Random loop models}
\label{sec loops}

Models of random loops have been introduced as representations of quantum spin systems \cite{Toth,AN,Uel}. They are increasingly popular in probability theory. A special example is the ``random interchange model'' where the outcomes are permutations given by products of random transpositions. We present a theorem concerning the decay of loop correlations that is plausible in the context of quantum spins, but is quite surprising in the probabilistic context.

To each edge of the graph $(\Lambda,\caE)$ is attached the ``time'' interval $[0,\beta]$. Independent Poisson point processes result in the occurrences of ``crosses'' with intensity $u$ and ``double bars'' with intensity $1-u$, where $u \in [0,1]$ is a parameter. This means that, on the edge $\{x,y\} \in \caE$ and in the infinitesimal time interval $[t,t+\dd t] \subset [0,\beta]$, a cross appears with probability $u \dd t$, a double bar appears with probability $(1-u) \dd t$, and nothing appears with probability $1 - \dd t$. We denote by $\rho$ the measure and by $\omega$ its realizations.

Given a realisation $\omega$, loops are closed trajectories of travelers traveling along the time direction, with periodic boundary conditions at $0$ and $\beta$, who usually rest on a site of $\Lambda$, but jump to a neighboring site whenever a cross or a double bar is present. If a cross is encountered, the trajectory continues in the same direction of the time axis; at a double bar, the trajectory continues in the opposite time direction; see the illustration in Fig.\ \ref{fig loops}. We let $\caL(\omega)$ denote the set of loops of the realization $\omega$. Notice that $|\caL(\omega)| < \infty$ with probability 1.

\bfig
\includegraphics{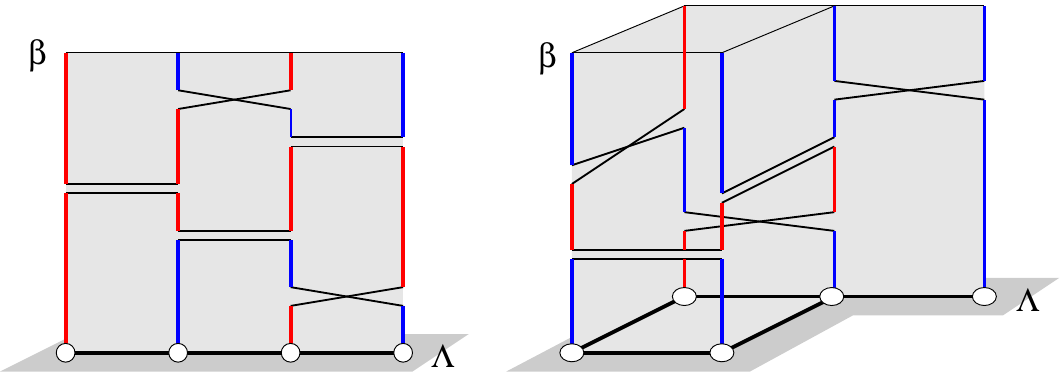}
\caption{Illustrations for the random loop models. The vertices all lie in the horizontal plane and random crosses and double bars occur in the ``time'' intervals $[0,\beta]$ on top of each edges. In both of these examples, the realizations have exactly two loops, denoted in red and blue.}
\label{fig loops}
\efig

The partition function of the model is given by
\be
\label{Z_u}
Z = \int \theta^{|\caL(\omega)|} \rho(\dd\omega),
\ee
where $\theta>0$ is some parameter.
The ``equilibrium" measure is defined by
\be
\mu(\dd\omega) = \frac{1}{Z} \theta^{|\caL(\omega)|} \rho(\dd\omega).
\ee
The special example where $u=1$ and $\theta=1$ is the random interchange model; crosses stand for transpositions, and the loops are equivalent to permutation cycles.

We will prove the following result on the probability, $\mathbb{P}(x\leftrightarrow y)$, of two sites,  $x,y$, to belong to the same loop.

\begin{theorem}
\label{thm loops}
Let $\theta=2,3,4,\dots$, $u \in [0,1]$ and $x,y\in\Lambda$. There exist positive constants $C>0$  and 
$\xi(\beta)>0$, the latter depending on $\beta$, $\gamma$, $\theta$, $u$ but not on $x,y\in\Lambda$, such that
\[
\bbP(x \leftrightarrow y) \leq C \left(d(x,y)+1\right)^{-\xi(\beta)}.
\]
The asymptotics of $\xi(\beta)$ for large values of $\beta$ is given by
\[
\lim_{\beta \rightarrow \infty} \beta \,\xi(\beta)= \bigl[ 8\gamma^2 (\theta-1)^2 (u+(1-u)\theta+1) \bigr]^{-1}.
\]
\end{theorem}

The proof is based on the correspondence that exists between this random loop model and certain quantum spin systems with continuous symmetry. This allows us to use the general result in Theorem \ref{thm general}, see Section \ref{sec proofs} for details. Such correspondence exists only for the values of $\theta$ specified in the theorem. We expect that the result holds for all $\theta>0$, though.

\section{The Hubbard Model}
\label{sec Hubbard}

Let $\Lambda$ be a finite region in a lattice. We define standard fermionic creation- and annihilation operators, $c^\dagger_{\sigma, x}$, $c_{\sigma, x}$, $x \in \Lambda$, $\sigma = \uparrow, \downarrow$, for spin-$\frac{1}{2}$ fermions. These operators act on a Hilbert space, $\mathcal{H}_{x}$, associated with site $x$ and defined by
$\mathcal{H}_x = \mbox{span}\{0, \uparrow,\downarrow,\uparrow\downarrow\}\simeq \bbC^4$. The creation and annihilation operators satisfy the usual anticommutation relations, $\left\{c_{\sigma,x}, c^\dagger_{\sigma',y}\right\}= \delta_{xy}\delta_{\sigma \sigma'}$. 

The Hubbard model is a model of electrons, which are spin-$\frac{1}{2}$ fermions, described in a tight-binding approximation. We consider a general family of such models, including ones with electron hopping amplitudes, $\lbrace t_{xy} \rbrace$, of long range. The hamiltonian is given by
\be
\label{HHM2}
H_\Lambda = -\frac{1}{2}\sum_{x,y\in\Lambda} \sum_{\sigma = \uparrow,\downarrow}t_{xy} \big( c^\dagger_{\sigma, x}c_{\sigma, y}+ c^\dagger_{\sigma,y}c_{\sigma,x} \big) + 
V\left(\{n_{\uparrow,x}\}_{x\in\Lambda} , \{n_{\downarrow,x}\}_{x \in \Lambda}\right) .
\ee
The number operators, $n_{\sigma, x}, \sigma = \uparrow, \downarrow$, are defined in the usual way: $n_{\sigma,x} = c^{\dagger}_{\sigma,x}c_{\sigma,x}$, and $n_x = n_{\uparrow,x} + n_{\downarrow,x}$. 
It is assumed that the hamiltonian of the model is only invariant
under rotations around the 3-axis in spin space, which form a U(1)-
symmetry group. Accordingly, the potential
$V(\{ n_{\sigma, x} \})$
is only assumed to depend on the occupation numbers $n_{\sigma, x}, \sigma \in
\lbrace\uparrow, \downarrow \rbrace$, $x \in \Lambda$; but no further assumptions
are needed.

Under the extra assumption that $V$ depends on $\{n_x\}_{x\in\Lambda}$ instead of $\{n_{\uparrow,x}\}_{x\in\Lambda} , \{n_{\downarrow,x}\}_{x \in \Lambda}$, the hamiltonian in \eqref{HHM2} exhibits an SU(2) symmetry, with generators
\be
\mathcal{S}_x^{+} = c^\dagger_{\uparrow,x}c_{\downarrow,x}, \qquad\mathcal{S}^-_x= (\mathcal{S}_x^+)^\dagger, \qquad \mathcal{S}_x^3=\tfrac{1}{2}(n_{\uparrow,x}-n_{\downarrow,x})\label{SU(2)a}.
\ee
For background on the Hubbard model, we recommend the excellent review \cite{Lieb}. 

The hamiltonian \eqref{HHM2} still enjoys two U(1) symmetries, which is enough for our purpose. The first symmetry corresponds to the conservation of the spin component along the third axis, namely
\be
\label{comm s3 Hubbard}
\Bigl[ \sum_{\sigma = \uparrow,\downarrow}t_{xy} \left(c^\dagger_{\sigma, x}c_{\sigma, y}+ c^\dagger_{\sigma,y}c_{\sigma,x} \right), \caS^3_x +\caS^3_y\Bigr]=0.
\ee
The second symmetry deals with the conservation of the number of particles:
\begin{equation}
\label{U(1) Hubbard}
\Bigl[\sum_{\sigma = \uparrow,\downarrow}t_{xy} \left(c^\dagger_{\sigma, x}c_{\sigma, y}+ c^\dagger_{\sigma,y}c_{\sigma,x} \right), \sum_{\sigma=\uparrow,\downarrow}\left(n_{\sigma,x}+n_{\sigma,y}\right)\Bigr]=0.
\end{equation}
Different symmetries can be used to estimate the decay of different correlation functions.
Specifically, we analyze three different two-point functions:
\begin{itemize}
\item[(i)] $\langle c^\dagger_{\uparrow, x}c_{\downarrow,x}c^\dagger_{\downarrow, y}c_{\uparrow,y}\rangle$, measuring magnetic long-range order;
\item[(ii)] $\langle c_{\uparrow,x}^\dagger c_{\downarrow,x}^\dagger c_{\uparrow,y} c_{\downarrow,y}\rangle$, related to Cooper pairs and superconductivity;
\item[(iii)] $\langle c^\dagger_{\sigma,x} c_{\sigma,y}\rangle$, measuring off-diagonal long-range order.
\end{itemize}
The latter two correlation functions have been studied in \cite{KT, MR}. In \cite{KT}, their decay is studied with the help of a method similar to ours, and under the condition that $t_{xy}=0$ if $d(x,y)\geq R$, for some positive $R$. In \cite{MR}, it is assumed that $t_{xy}$ decays rather rapidly, more precisely, $t_{xy}\sim t\,d(x,y)^{-\alpha}$, with $\alpha>4$ and $t$ some constant. We will see that we have to require the same conditions in order for the general result in Theorem \ref{thm general} to be applicable.

\begin{theorem}
\label{thm HHM}
Let $H_{\Lambda}$ be the hamiltonian of the Hubbard model \eqref{HHM2} defined on the lattice $\Lambda$, and $x,y\in\Lambda$. Suppose that $t_{xy}= t (d(x,y)+1)^{-\alpha}$ with $\alpha>4$. Then there exist $C>0$, $\xi(\beta)>0$ (the latter depending on $\beta$, $\gamma$, $\alpha$, $t$, but not on $x,y\in\Lambda$) such that
\[
\left.
\begin{array}{c}
|\langle c^\dagger_{\uparrow, x}c_{\downarrow,x}c^\dagger_{\downarrow, y}c_{\uparrow,y}\rangle|\\
 |\langle c_{\uparrow,x}^\dagger c_{\downarrow,x}^\dagger c_{\uparrow,y} c_{\downarrow,y}\rangle|\\
|\langle c_{\sigma,x}^\dagger  c_{\sigma,y}\rangle|
\end{array}
\right\}\leq C (d(x,y)+1)^{-\xi(\beta)}
\]
where $\sigma\in\{\uparrow,\downarrow\}$ in the last line. Furthermore,
\[
\lim_{\beta\rightarrow\infty}\beta\,\xi(\beta)= \Bigl( 64\gamma ^2 |t| \sum_{r\geq 1} r^{-\alpha + 3} \Bigr)^{-1}.
\]
\end{theorem}

We note that this theorem asserts that there is a power-law upper bound on the decay of correlation functions, provided 
$\alpha>4$. As explained in the proof (Section \ref{sec proofs}), this condition is necessary to ensure finiteness of the $K$-norm  of the interaction (see Eq.\ \eqref{norm}) independently of the size of $\Lambda$.

\section{The $t$-$J$ model}
\label{sec t-J}

A well known variant of the Hubbard model is given by the $t$-$J$ model. The hamiltonian of this model is given by 

\begin{equation}
H_\Lambda = -\frac{t}{2}\sum_{\langle x,y\rangle\in\mathcal{E}}\sum_{\sigma = \uparrow,\downarrow}\left(c^\dagger_{\sigma, x}c_{\sigma, y}+c^\dagger_{\sigma,y}c_{\sigma,x}\right) +J\sum_{\langle x,y\rangle\in\mathcal{E}}\Bigl(\vec{\mathcal{S}}_{x}\cdot \vec{\mathcal{S}}_{y}-\frac{1}{4} n_{x}n_{y} \Bigr).
\label{HtJM}
\end{equation}
The parameters $t$ and $J $ are real numbers, and $\mathcal{S}^i_x = \sum_{\sigma,\mu = \uparrow,\downarrow} c^\dagger_{\sigma,x} \tau^{i}_{\sigma,\mu} c_{\mu,x}$, with $i = \{1,2,3\}$, and $\tau^1$, 
$\tau^2$, $\tau^3$ are the three Pauli matrices for particles of spin $\frac{1}{2}$. Explicitly,
\be
\begin{split}
\mathcal{S}^1_{x} =& \tfrac{1}{2} \bigl( c^\dagger_{\uparrow,x}c_{\downarrow x} + c^\dagger_{\downarrow,x}c_{\uparrow y} \bigr),\\ 
\mathcal{S}^2_{x} =& - \tfrac{\ii}{2} \bigl( c^\dagger_{\uparrow,x}c_{\downarrow x} - c^\dagger_{\downarrow,x}c_{\uparrow y} \bigr),\\ 
\mathcal{S}^3_{x} =& \tfrac{1}{2}(n_{\uparrow,x}-n_{\downarrow,x}).
\end{split}
\ee
These are the generators of a representation of the symmetry group SU(2) on the state space of the model, as previously introduced for the Hubbard model (see Eq.\ \eqref{SU(2)a}). In the $t$-$J$ model the number of particles and the component of the total spin along, for example, the third axis are conserved -- i.e., the model exhibits two $U(1)$ symmetries: For all $x,y \in \Lambda$,
\be
\label{comm tJ}
\begin{split}
&\Bigl[ -\tfrac{t}{2}\sum_{\sigma = \uparrow,\downarrow}\left(c^\dagger_{\sigma, x}c_{\sigma, y}+c^\dagger_{\sigma,y}c_{\sigma,x}\right) +J \bigl( \vec{\mathcal{S}}_{x}\cdot \vec{\mathcal{S}}_{y}-\tfrac{1}{4} n_{x}n_{y} \bigr),\sum_{\sigma = \uparrow,\downarrow} n_{\sigma,x}+n_{\sigma,y} \Bigr] = 0 \\
&\Bigl[ -\tfrac{t}{2}\sum_{\sigma = \uparrow,\downarrow}\left(c^\dagger_{\sigma, x}c_{\sigma, y}+c^\dagger_{\sigma,y}c_{\sigma,x}\right) +J \bigl( \vec{\mathcal{S}}_{x}\cdot \vec{\mathcal{S}}_{y}-\tfrac{1}{4} n_{x}n_{y} \bigr),\\
&\hspace{6.4cm}\frac{1}{2} (n_{\uparrow,x} - n_{\downarrow,x}+n_{\uparrow,y}-n_{\downarrow,y}) \Bigr] = 0.
\end{split}
\ee

The analysis carried out in the Hubbard model holds for the $t$-$J$ model too, and yields bounds on the decay of various correlation functions.

\begin{theorem}\label{thm HtJM}
Let $H_{\Lambda}$ be the hamiltonian defined in \eqref{HtJM}, and let $x,y$ be two sites of the lattice $\Lambda$.
Then there exist constants $C>0$ and $\xi(\beta)>0$  (the latter depending on $\beta$, $\gamma$, $t$ and $J$, but not on $x,y\in\Lambda$) such that
\[
\left.\begin{array}{c}
|\langle c_{\uparrow,x}^\dagger c_{\downarrow,x} c_{\downarrow,y}^\dagger c_{\uparrow,y}\rangle\\
|\langle c_{\uparrow,x}^\dagger c_{\downarrow,x}^\dagger c_{\uparrow,y} c_{\downarrow,y}\rangle|\\
|\langle c_{\sigma,x}^\dagger  c_{\sigma,y}\rangle| 
\end{array}\right\}
\leq C\left(d(x,y)+1\right)^{-\xi(\beta)}
\]
with $\sigma\in\{\uparrow,\downarrow\}$. Furthermore,
\[
\lim_{\beta\rightarrow\infty}\beta\,\xi(\beta) = (128 \gamma^2\left(2|t|+|J|\right))^{-1}.
\]
\end{theorem}

\section{General model with U(1) symmetry}
\label{sec general}

Let $(\Lambda,\caE)$ denote a finite graph, with $\Lambda$ the set of vertices and $\caE$ the set of edges, whose perimeter constant $\gamma$ is finite; see Eq.\ \eqref{def gamma}. Let $\caH_{\Lambda}$ be a finite-dimensional Hilbert space. Typically, $\mathcal{H}_{\Lambda}$ is given by a tensor product $\otimes_{x\in\Lambda} \bbC^{N}$, but we will not make use of this special structure. Let $\caB(\caH_{\Lambda})$ denote the algebra of linear operators on $\caH_{\Lambda}$. We assume that there are sub-algebras $\caB_{A} \subset \caB(\caH_\Lambda)$, with the properties that ${\bf{1}} \in \mathcal{B}_{A}$, for all $A \subset \Lambda$, and $\mathcal{B}_{A} \subseteq \mathcal{B}_{A'}$, whenever 
$A\subseteq A' \subseteq \Lambda$, and hermitian operators $(S_{x})_{x\in\Lambda}$ 
obeying the following commutation relations:
\begin{itemize}
\item[(a)] For arbitrary $x,y \in \Lambda$, with $x \neq y$, we have that
\be
\label{condition1}
[S_{x},S_{y}] = 0.
\ee
\item[(b)] For arbitrary $x\in\Lambda$, $A\subset\Lambda$, and $\Psi \in \caB_{A}$,
\be\label{condition2}
[S_{x}, \Psi] \begin{cases} = 0 & \text{if } x \notin A; \\ \in \caB_{A} & \text{if } x \in A. \end{cases}
\ee
\end{itemize}

The hamiltonian of the model is a sum of ``local'' interactions. More precisely, we assume that
\begin{equation}
\label{Ham}
H_{\Lambda} = \sum_{A\subset\Lambda} \Phi_A,
\end{equation}
where the operator $\Phi_{A}$ is hermitian and belongs to $\caB_{A}$, for all $A \subset \Lambda$. This hamiltonian is assumed to be invariant under a U(1) symmetry with generator $\sum_{x} S_{x}$, in the precise sense that
\begin{equation}
\label{symm}
\Bigl[ \Phi_A, \sum_{x\in A} S_x \Bigr]= 0,
\end{equation}
for all $A \subset \Lambda$. Without loss of generality, we assume that 
\be
\label{norm S}
\|S_x\| = 1, \forall x \in \Lambda.
\ee
We introduce a norm on the space of interactions, $\Phi_{\cdot}$, depending on a parameter $K\geq 0$; namely
\be
\|\Phi\|_K = \sup_{y\in\Lambda}\sum_{\substack{A\subset\Lambda\\{\rm s.t. }\,y\in A}}\|\Phi_A\|(|A|-1)^2({\rm diam}(A)+1)^{2K(|A|-1)+2}.
\label{norm}
\ee
Notice that this $K$-norm does not depend on possible ``one-body terms'', ($\vert A \vert =1$).

As usual, the Gibbs state $\langle (\cdot) \rangle$ is the positive, normalized linear functional that assigns the expectation value
\begin{equation}
\langle a\rangle = \frac{\Tr a \e{-\beta H_\Lambda}}{\Tr \e{-\beta H_{\Lambda}}}
\end{equation}
to each operator $a \in \caB(\caH_{\Lambda})$.

Next, we assume that there exists a ``correlator'' $O_{xy} \in \caB_\Lambda$, for some $x,y \in \Lambda$, satisfying the following commutation relation: There is a constant $c \in \bbR$ such that
\be
\label{comm}
[S_x,O_{xy}] = cO_{xy},  \text{   and   } [S_{z}, O_{xy}] = 0, \text{   for   } z\not= x,y.
\ee
Notice that there are no assumptions about the commutator between $S_y$ and $O_{xy}$.
We are now prepared to state a general version of the McBryan-Spencer-Koma-Tasaki theorem \cite{MS,KT}, claiming power-law decay of certain two-point functions for the general class of models introduced above.

\begin{theorem}
\label{thm general} 
Suppose that the constant $\gamma$ in Eq.\ \eqref{def gamma} is finite, and that $\{S_{x}\}_{x\in\Lambda}$, $(\Phi_{A})_{A\subset\Lambda}$, and $O_{xy}$ satisfy properties \eqref{condition1}--\eqref{norm S} and \eqref{comm}. Then there exist $C>0$ and $\xi(\beta)>0$ (uniform with respect to $\Lambda$ and $x,y\in\Lambda$) such that
\[
|\langle O_{xy}	\rangle| \leq C \left(d(x,y)+1\right)^{-\xi(\beta)}.
\]
Moreover, if there exists a positive constant $K$ such that $\|\Phi\|_{K}$ is bounded uniformly in $\Lambda$, then
\[
\lim_{\beta\rightarrow\infty}\beta \,\xi(\beta) = \frac{c^2}{8\gamma\|\Phi\|_0}.
\]
\end{theorem}

In the remainder of this section we present a proof of this theorem. We follow the method of Koma and Tasaki, which they developed in the context of the Hubbard model \cite{KT}. As in \cite{FU}, we use the H\"older inequality for traces, which simplifies the proof as compared to \cite{KT}.

\begin{proof}
The proof is based on a use of ``complex rotations'', as first introduced in \cite{MS}.
We define (imaginary) ``rotation angles'', $\{\theta_z\}_{z\in\Lambda}\in\bbR^\Lambda$, as follows:
\be\label{theta}
\theta_z =\left\{\begin{array}{cl}\kappa \log\frac{d(x,y)+1}{d(x,z)+1}&{\rm if }\;d(x,z)\leq d(x,y), \\ 0&{\rm otherwise,}\end{array}\right.
\ee
where $\kappa$ is an arbitrary positive parameter that will be used to optimize our bounds.
An operator of complex rotations is defined by
\begin{equation}
R = \prod_{z\in\Lambda} \mbox{e}^{\theta_{z} S_{z}}.
\end{equation}
For each set $A\subseteq \Lambda$, we let $x_{0}(A)$ be the site (or one of the sites) in $A$ that has minimal (Manhattan) distance from $x$. Using \eqref{symm}, we have that
\be
\label{rot1}
\begin{split}
R^{-1} H_\Lambda R &=  \sum_{A\subset\Lambda}\mbox{e}^{-\sum_{z\in A}(\theta_z-\theta_{x_0(A)})S_z}\Phi_A \,\mbox{e}^{\sum_{z\in A}(\theta_z-\theta_{x_0(A)})S_y} \\
&= \e{-T_{A}} \Phi_{A} \e{T_{A}},
\end{split}
\ee
where
\be
T_{A} = \sum_{z \in A} (\theta_{z}-\theta_{x_{0}}) S_{z}.
\ee
Recall the notation ${\rm ad}_{a}(b) := [a,b]$. We use the multi-commutator expansion to show that
\be
\begin{split}
R^{-1}H_\Lambda R &=
\sum_{A\subset\Lambda}\Phi_A +\sum_{j\geq 1}\sum_{A\subset\Lambda} \frac{(-1)^j}{j!} \, \mbox{ad}_{T_{A}}^j(\Phi_A) \\
&= H_\Lambda + B + C,
\end{split}
\ee
where
\be
B = -\sum_{\substack{j\geq 1}}\sum_{A\subset\Lambda} \frac{1}{(2j-1)!}\mbox{ad}_{T_{A}}^{2j-1}(\Phi_A),
\ee
and
\be
C = \sum_{\substack{j\geq 1\\}}\sum_{A\subset\Lambda} \frac{1}{(2j)!}\mbox{ad}_{T_{A}}^{2j}(\Phi_A).
\ee
The operator $B$ contains all terms of the multicommutator expansion odd in $T_{A}$ and is therefore anti-hermitian; $C$ contains the terms even in $T_{A}$ and, hence, is hermitian. 
Eqs.\eqref{comm} and \eqref{theta} imply that
\be\label{rotO}
R^{-1} O_{xy} R = \mbox{e}^{-c\theta_x}O_{xy}.
\ee
Next, we apply the Trotter formula and the H\"older inequality for traces of products of matrices, to find that
\be
\begin{split}
\left|\Tr \big( O_{xy}\mbox{e}^{-\beta H_\Lambda} \big) \right| &= \left|\Tr \big( R^{-1}O_{xy}R\mbox{e}^{-\beta R^{-1}H_\Lambda R} \big) \right| \\
&= \mbox{e}^{-c\theta_x}\left|\Tr  \big( O_{xy}\mbox{e}^{-\beta H_\Lambda -\beta B -\beta C} \big) \right| \\
&\leq \mbox{e}^{-c\theta_x}\lim_{n\rightarrow \infty}\left|\Tr \Big( O_{xy} \big(\mbox{e}^{-\frac{\beta}{n} H_\Lambda}\mbox{e}^{-\frac{\beta}{n} B}\mbox{e}^{ -\frac{\beta}{n}C}\big)^n \Big) \right| \\
&\leq \mbox{e}^{-c\theta_x} \lim_{n\rightarrow \infty} \|O_{xy}\|_\infty\| \,\mbox{e}^{-\frac{\beta}{n}H_\Lambda}\|_n^n \, \|\mbox{e}^{-\frac{\beta}{n}B}\|^n_{\infty} \, \|\mbox{e}^{-\frac{\beta}{n}C}\|_\infty^n.
\label{traccia1}
\end{split}
\ee
Notice that $\|\mbox{e}^{-\frac{\beta}{n}B}\|_{\infty} = 1$, because $B$ is anti-hermitian. 

Moreover, 
$\|\mbox{e}^{-\frac{\beta}{n}H_\Lambda}\|_n^n  = \Tr\e{-\beta H_{\Lambda}}$, so that, using Eq.\ \eqref{theta},
\begin{equation}
\left|\langle O_{xy}\rangle\right| \leq \e{-c\kappa \log (d(x,y)+1)} \|O_{xy}\| \mbox{e}^{\beta \|C\|}.
\end{equation}

Next, we estimate $\|C\|$. Using $\|[A,B]\| \leq 2 \|A\| \, \|B\|$, we obtain
\be
\begin{split}
\|C\|&\leq\sum_{j\geq 1} \sum_{A\subset \Lambda}\frac{2^{2j}}{(2j)!}\|T_A\|^{2j}\|\Phi_A\|\\&=\sum_{A\subset\Lambda}\|\Phi_A\| \bigl( \cosh\left(2\|T_A\|\right)-1 \bigr).
\label{normC}
\end{split}
\ee
Using the inequality $\cosh u - 1 \leq \frac12 u^2 \e{u}$, which is easily verified for all $u\geq0$, we find that
\be
\begin{split}
&\|C\|\leq 2\sum_{A\subset\Lambda}\|\Phi_A\| \|T_A\|^2 \e{2\|T_A\|}\\
&\leq2 \sum_{A\subset\Lambda}\|\Phi_A\|\left(|A|-1\right)\e{2\sum_{z\in A\backslash \{x_0(A)\}}|\theta_z -\theta_{x_0(A)}| }\sum_{z\in A\backslash \{x_0(A)\}}|\theta_z-\theta_{x_0(A)}|^2.
\end{split}
\ee
Here a bound on $\|T_A\|^2$ has been used that follows from the Cauchy-Schwarz inequality; namely
\be
\|T_A\|^2 \leq \Bigl( \sum_{y \in A \setminus \{x_0\}} |\theta_y - \theta_{x_0}| \Bigr)^2 \leq (|A|-1) \sum_{y \in A} |\theta_y - \theta_{x_0}|^2.
\ee
From the explicit form of the ``angles'' $\{\theta_z\}_{z\in\Lambda}$ displayed in Eq.\ \eqref{theta} it then follows that
\be
\begin{split}
&\|C\|\leq 2\sum_{\substack{A\subset\Lambda\;{\rm s.t.}\\d(x,x_0(A))\leq d(x,y)}} \|\Phi_A\| (|A|-1)({\rm diam}(A) +1)^{2\kappa(|A|-1)} \\
& \hspace{5cm} \cdot \sum_{z\in A\backslash \{x_0(A)\}}|\theta_z - \theta_{x_0(A)}|^2\\
&\leq2\kappa^2\sum_{\substack{A\subset\Lambda\;{\rm s.t.}\\d(x,x_0(A))\leq d(x,y)}} \|\Phi_A\|(|A|-1)^2({\rm diam}(A)+1)^{2\kappa(|A|-1)+2} \\
& \hspace{5cm} \cdot \frac{1}{(d(x,x_0(A))+1)^2}.
\end{split}
\ee
We further estimate $\|C\|$ by reorganizing the sums and using definition \eqref{norm} of the norm 
$\|\Phi\|_\kappa$. Thus
\be
\begin{split}
&\|C\|\leq 2\kappa^2\sum_{\substack{x_0\in\Lambda\;{\rm s.t.}\\d(x,x_0)\leq d(x,y)}}\frac{1}{(d(x,x_0)+1)^2}\sum_{A\ni x_0}\|\Phi_A\|(|A|-1)^2\\&\hspace{5cm}\cdot({\rm diam}(A)+1)^{2\kappa(|A|-1)+2}\\
&\leq 2\kappa^2\sum_{\substack{x_0\in\Lambda\;{\rm s.t.}\\d(x,x_0)\leq d(x,y)}}\frac{1}{(d(x,x_0)+1)^2}\|\Phi\|_\kappa.
\end{split}
\ee
Recall the definition of the perimeter constant $\gamma$ in Eq.\ \eqref{def gamma}. Since we only consider graphs $(\Lambda,\caE)$ for which $\gamma$ is finite, we have
\be
\begin{split}
\|C\|&\leq 2\kappa^2\|\Phi\|_\kappa \biggl( \sum_{r=1}^{d(x,y)}\frac{\gamma r}{(1+r)^2}+1 \biggr) \\
&\leq2\kappa^2|\Phi\|_{\kappa} \biggl( \sum_{r=1}^{d(x,y)}\frac{\gamma}{r}+1 \biggr) \\
&\leq2\kappa^2\gamma\|\Phi\|_\kappa\log(d(x,y)+1)+ 2\kappa^2\|\Phi\|_\kappa.
\end{split}
\ee
We conclude that, for all $\kappa >0$,
\be\label{upperbd}
|\langle O_{xy}\rangle|\leq \caC_\kappa (d(x,y)+1)^{-(\kappa c -2\kappa^2\gamma\|\Phi\|_\kappa \beta)}
\ee
where $\caC_\kappa = \|O_{xy}\| \e{2\kappa^2 \|\Phi\|_\kappa}$ and $c$ is the constant defined in equation\eqref{comm} and used in equation \eqref{rotO}. 
Next, we verify that the exponent on the right side of \eqref{upperbd} is
$\propto \frac{1}{\beta}$, for $\beta$ large enough. Choosing $\kappa = \frac{K}{\beta}$, this exponent is given by
\be
\xi_K(\beta)= \tfrac{K}{\beta} \bigl( c -2K\gamma\|\Phi\|_{\frac{K}{\beta}} \bigr).
\label{xi_K}
\ee
Recall that in the last part of Theorem \ref{thm general} it is assumed that there is a constant $\tilde{K}$ such that the $\tilde{K}$-norm of the interaction $\Phi$ converges, independently of $\Lambda$. Applying dominated convergence to $\|\Phi\|_K$, we then get
\be
\lim_{\beta\rightarrow\infty}\beta\, \xi_K(\beta) = Kc-2K^2\gamma \|\Phi\|_0.
\ee
The optimal value of $K$ is $K^{*}= c/(4\gamma\|\Phi\|_0)$. We define $\xi(\beta)= \xi_{K^*}(\beta)$, and substitute $\caC_\kappa$ with $\caC_{\frac{K^*}{\beta}}$ in Eq.\ \eqref{upperbd}. This completes the proof.
\end{proof}

\section{Applications of the general theorem to the explicit examples}
\label{sec proofs}

In this section we sketch the proofs of the theorems stated in Sections \ref{sec spins}--\ref{sec Hubbard}. They are all straightforward applications of Theorem \ref{thm general}.

\begin{proof}[Proof of Theorem \ref{thm HSM}]

The interaction defining the hamiltonian has finite $K$-norm, for any $K>0$.
\be
\|\Phi\|_K = 2^{2K +2} \sup_{z\in\Lambda}\sum_{w\sim z} \Bigl\| \sum_{l = 1}^{2s} c_l(w,z) (\vec{\mathcal{S}}_w\cdot\vec{\mathcal{S}}_z)^l \Bigr\| \leq 2^{2K+2}\gamma.
\label{boundnormHSM}
\ee
The bound follows from the triangular inequality and the assumption in Eq.\ \eqref{assumptionHSM}. Let $S_x =\frac{1}{s}\mathcal{S}^3_x$. It is bounded with norm 1 and $S_x + S_y$ commutes with the local hamiltonian so it provides the U(1) symmetry of Eq.\ \eqref{comm spin}. Let $O_{xy} = \mathcal{S}^+_x\mathcal{O}_y$ for some $\mathcal{O}_y\in\caB_y$. It is bounded and
\begin{equation}
\left[S_x, O_{xy}\right]=s^{-1}O_{xy}.
\end{equation}
Then, the value of $c$ as defined in Theorem \ref{thm general}, Eq.\ \eqref{comm}, is $c=s^{-1}$.
The result is now a straightforward application of Theorem \ref{thm general}. Consider $\xi_K(\beta)$ as defined in the proof of the general Theorem \ref{thm general}, Eq.\ \eqref{xi_K}. 
We get from Eq.\ \eqref{boundnormHSM}
\be
\xi_K(\beta)\geq\tfrac{K}{\beta} \bigl( s^{-1} - 8K^2\gamma^2 2^{\frac{2K}\beta}\bigr)=\tilde{\xi}_K(\beta).
\ee
It is clear that $\lim_{\beta\rightarrow\infty} \beta\, \tilde{\xi}_K (\beta) =\frac{K}s-8K^2\gamma^2$. By optimising with respect to $K$, we get the first statement of the theorem by defining $\xi(\beta) = \tilde{\xi}_{K^*}(\beta)$ where $K^*$ is the optimal value of $K$.

Using the SU(2) invariance of the Gibbs state of the model, which follows from Eqs. \eqref{comm spin} and \eqref{HSM}, we see that
\begin{equation}
\langle \mathcal{S}^1_x \mathcal{S}^1_y \rangle = \langle \mathcal{S}^2_x \mathcal{S}^2_y \rangle = \langle \mathcal{S}^3_x \mathcal{S}^3_y \rangle.
\end{equation}
The definition of $\mathcal{S}^+_x$ and $\mathcal{S}^-_x$ then implies that
\begin{equation}
\langle \mathcal{S}^+_x \mathcal{S}^-_y \rangle = 2\langle \mathcal{S}^1_x \mathcal{S}^1_y\rangle = 2\langle \mathcal{S}^2_x \mathcal{S}^2_y \rangle.
\end{equation}
The second claim in Theorem \ref{thm HSM} is a special case of the first one.
\end{proof}

\begin{proof}[Proof of Theorem \ref{thm loops}]
For $\theta$ an integer larger than 1, the loop model is equivalent to a quantum spin model \cite{Toth, AN, Uel}. Let $\theta = 2s + 1$ with $s \in \frac{1}{2}\bbN$.
We introduce operators acting on $\bbC^{\theta}\otimes \bbC^{\theta}$, namely
\be
\begin{split}
&T\, e_i\otimes e_j = e_i\otimes e_j,\\
&(e_i\otimes e_j, Q \,e_l\otimes e_k) =\delta_{i,j}\delta_{l,k},
\end{split}
\ee
where $\{e_j\}_{j=1}^\theta$ denotes the canonical basis of $\bbC^\theta$.
Then we consider the Hilbert space $\caH_\Lambda = \otimes_{x\in\Lambda}\bbC^{\theta}$, and the hamiltonian 
\begin{equation}\label{ham}
H_\Lambda = -\sum_{\{x,y\}\in\caE}\left(uT_{xy}+(1-u)Q_{xy}-1\right).
\end{equation}
Here, $T_{xy}$ stands for $T\otimes\mathds{1}_{\Lambda\backslash\{x,y\}}$, and similarly for $Q_{xy}$.  The corresponding Gibbs state at inverse temperature $\beta$ is $\langle\cdot\rangle=\Tr \cdot \e{-\beta H_\Lambda}/\Tr \e{-\beta H_\Lambda}$.
It can be shown that, for any $u\in [0,1]$, the partition function $Z$ in Eq.\ \eqref{Z_u} is equal to the quantum partition function, namely
\be
Z = \Tr_{\caH_\Lambda} \e{-\beta H_\Lambda}.
\ee

It can be shown \cite{Uel} that for any $x,y\in\Lambda$,
\be
\begin{split}
&[T_{xy}, \caS^i_x+\caS^i_y] = 0, \quad i=1,2,3\\
&[Q_{xy}, \caS^2_x+\caS^2_y] = 0.
\end{split}
\ee
This implies that we can set $S_x:= \frac{1}{s}\caS^2_x$, which has the right commutation relations with the hamiltonian in \eqref{ham}. Moreover, it is easy to check that the interaction has finite
 K-norm, for any $K>0$:
\be
\|\Phi\|_K = 2^{2K+2}\sup_{y\in\Lambda}\sum_{x\sim y} \| u T_{xy} + (1-u) Q_{xy}-1\|\leq 2^{2K+2}\gamma\left(u+(1-u)\theta+1\right).\label{boundnormloops}
\ee
The bound follows from the triangular inequality and from $\|T\|=1$, $\|Q\|=\theta$.
Let $\mathcal{Q}^\pm = \caS^1\pm i\caS^3$.
Then for any $x,y\in\Lambda$ and $\mathcal{O}\in\mathcal{B}_y$,
\be
\left[ S_x, \mathcal{Q}^+_x\mathcal{O}_y\right] = s^{-1}\mathcal{Q}^+_x \mathcal{O}_y,
\ee
i.e. $c=s^{-1}$ in Theorem \ref{thm general}.

The bounds in Theorem \ref{thm general} can be applied to the correlator 
$\langle \mathcal{Q}^+_x\mathcal{Q}^-_y\rangle$. Indeed, let $\xi_K(\beta)$ be as defined in the proof of Theorem \ref{thm general}, Eq.\ \eqref{xi_K}. From Eq.\ \eqref{boundnormloops}, we have
\be
\xi_K(\beta)\geq \frac{K}{\beta}\left(\frac{2}{\theta-1}-8K\gamma^2 2^{\frac{2K}{\beta}}(u+(1-u)\theta+1)\right)=\tilde{\xi}_K(\beta).
\ee
Moreover, we have that
\be
\lim_{\beta\rightarrow\infty}\beta\,\tilde{\xi}_K(\beta) =\frac{2K}{\theta-1}-8 K^2\gamma^2\left(u+\theta(1-u)+1\right).
\ee
Optimizing in $K$ and definining $\xi(\beta) = \tilde{\xi}_{K^*}(\beta)$, where $K^*$ is the optimal value of 
$K$, one obtains the result for $\langle \mathcal{Q}^+_x\mathcal{Q}^-_y\rangle$.

Due to the symmetry of the model,
\be
\langle \mathcal{S}^1_x\mathcal{S}^1_y\rangle =  \langle \mathcal{S}^3_x\mathcal{S}^3_y\rangle.
\ee
By the definition of $\mathcal{Q}^\pm$, we then find that
\be
\langle \mathcal{Q}^+_x\mathcal{Q}^-_y\rangle =2\langle \mathcal{S}^1_x\mathcal{S}^1_y\rangle  = 2\langle \mathcal{S}^3_x\mathcal{S}^3_y\rangle .
\ee
Thus, the result is proven for the correlation functions $\langle \caS^3_x\caS^3_y\rangle$.
The statement concerning the probability of two sites being connected follows from
\be
\langle \caS^3_x\caS^3_y\rangle=\frac{1}{12}(\theta^2-1)\mathbb{P}(x\leftrightarrow y).
\ee
See \cite{Uel} for a proof of this statement.
\end{proof}

\begin{proof}[Proof of Theorem \ref{thm HHM}]

It can be easily checked that the $K$-norm of the interaction associated to the hamiltonian is
\be
\label{finite norm}
\begin{split}
\|\Phi\|_K &= \sup_{x\in\Lambda}\sum_{y\in\Lambda} \Bigl\| \sum_\sigma c^\dagger_{\sigma,x}c_{\sigma,y} +c^\dagger_{\sigma,y}c_{\sigma,x} \Bigr\| \frac{|t|}{2} (d(x,y)+1)^{-(\alpha-2K-2)}.
\end{split}
\ee
First, notice that the one-body potential $V(\{n_{\uparrow,x}\}_{x\in\Lambda},\{n_{\downarrow,x}\}_{x\in\Lambda})$ does not play any role. 
Second, we would like the norm of the interaction to be independent of $\Lambda$, i.e., to be finite no matter what the size of $\Lambda$ is. By the triangular inequality and given the definition of $\gamma$,
\be
\|\Phi\|_K \leq 2|t|\gamma\sum_{r\geq1}r^{-(\alpha - 2K-3)}.
\ee
We note that, for any $\alpha>4$, there exists $K>0$ such that $\alpha>2K+4$. This ensures the existence of positive values of $K$ with the property that $\|\Phi\|_K$ is uniformly bounded in the size of $\Lambda$, as required in Theorem \ref{thm general}.

Let us focus our attention on the first two-point function.
We set $S_z =\left(n_{\uparrow,z} - n_{\downarrow, z}\right)$ $=2\mathcal{S}^3_z$, according to Eq.\ \eqref{SU(2)a}. These are bounded operators commuting with the hamiltonian, because the Hubbard hamiltonian exhibits a U(1) invariance corresponding to the conservation of the component of the total spin along the third axis; see Eq.\ \eqref{comm s3 Hubbard}.
Let $O_{xy}=c^\dagger_{\uparrow,x}c_{\downarrow,x}\mathcal{O}_y$, with $\mathcal{O}_y\in\caB_y$. Then
\begin{equation}
\left[S_x, O_{xy}\right]=2O_{xy}, \,\,\text{   and   }\,\,[S_z,O_{xy}]= 0,\,\,\forall z \not= x,y.
\end{equation}
Thus the constant $c$ in Theorem \ref{thm general} is given by $c=2$.

Next, we study the second correlator.
As seen in Section \ref{sec Hubbard}, the hamiltonian exhibits a U(1)- symmetry, thanks to the conservation of the number of particles; see Eq.\ \eqref{U(1) Hubbard}. Hence we can choose $S_x =\frac{1}{2}\left( n_{\uparrow,x}+n_{\downarrow,x}\right)$. 

Let $O_{xy} = c^\dagger_{\uparrow,x}c^\dagger_{\downarrow,x}\mathcal{O}_y$ for some $\mathcal{O}_y\in\mathcal{B}_y$. Then
\be
\left[S_x, O_{xy}\right]=O_{xy}, \quad [S_z, O_{xy}]=0, \,\, \forall z\not= x,y
\ee
and $c=1$ in Theorem \ref{thm general}. 

In the analysis of the third correlator, we also choose $S_z = \frac{1}{2}(n_{\uparrow,z}+n_{\downarrow, z})$.
Let $O_{xy} = c^\dagger_{\sigma,x}\mathcal{O}_y$, for an arbitrary $\sigma$ and an arbitrary operator
$\mathcal{O}_y\in\mathcal{B}_y$. Then
\be
[S_x, O_{xy}]=\frac{1}{2}O_{xy}, \,\, \text{ and }\,\, [S_z, O_{xy}]=0, \,\, \forall z \not= x,y.
\ee

The theorem is now a straightforward application of Theorem \ref{thm general}.
Indeed, let $\xi_K(\beta)$ be defined as in Eq.\ \eqref{xi_K}. In all three examples,
\be
\xi_K(\beta)\geq \frac{K}{2\beta}\left(1-8K\gamma^2 |t| \tau^{\frac{2K}{\beta}}\right) = \tilde{\xi}_K(\beta),
\ee
with $\tau^{\frac{2K}{\beta}}=\sum_{r\geq 1}r^{-\alpha+3 +\frac{2K}{\beta}}$, as is easily checked.
By dominated convergence,
\be
\lim_{\beta\rightarrow\infty} \beta\,\tilde{\xi}_K(\beta) = \frac{K}{2}(1-8K\gamma^2|t|\tau^0).
\ee
Optimizing in $K$ and defining $\xi(\beta)=\tilde{\xi}_K^*(\beta)$, where $K^*$ is the optimal value of $K$, the theorem follows, (after choosing $\mathcal{O}_y = c^\dagger_{\downarrow,y}c_{\uparrow,y}$ in the first case, $c_{\uparrow,y}c_{\downarrow,y}$  in the second case, and $ c_{\sigma,y}$ in the third case).

Notice that $\alpha>4$ is needed for $\tau^0$ to be well defined.
\end{proof}

\begin{proof}[Proof of Theorem \ref{thm HtJM}] 
The interaction defining the $t$-$J$ model has finite $K$-norm for any value of $K$ and it can be explicitly evaluated:
\be
\|\Phi\|_K =2^{2K+2}\sup_{x\in\Lambda}\sum_{x\sim y} \Bigl\|-\tfrac{t}{2}\sum_{\sigma}(c^\dagger_{\sigma,x}c_{\sigma_y}+c^\dagger_{\sigma,y}c_x)+J \bigl( \vec{\caS}_x\cdot\vec{\caS}_y-\tfrac{1}{4}n_x n_y\bigr)\Bigr\|.
\ee
We can bound $\|\Phi\|_K$ using the triangular inequality; by the definition of $\gamma$,
\be
\|\Phi\|_K\leq 2^{2K+2}\gamma\left(2|t|+|J|\right)\label{boundnormHtJM}.
\ee

To bound the first correlator, we set $S_z := \left(n_{\uparrow,z}-n_{\downarrow,z}\right)$. This operator commutes with the hamiltonian; see the second equation in \eqref{comm tJ}. Moreover, it is bounded, with norm equal to 1.
Let $O_{xy} = c^\dagger_{\uparrow,x}c_{\downarrow,x}\mathcal{O}_y$ with $\mathcal{O}_y\in\caB_y$. Then Eq. \eqref{comm} holds with $c=2$ (see Theorem \ref{thm general}).

To deal with the second correlator, we set $S_z := n_{\uparrow,z}+ n_{\downarrow,z}$. The hamiltonian conserves the number of particles, which corresponds to the U(1)- symmetry in the first equation of \eqref{comm tJ}. Thus, this operator commutes with the hamiltonian.
Let  $O_{xy}= c^\dagger_{\uparrow,x}c^\dagger_{\downarrow,x}\mathcal{O}_y$ with $\mathcal{O}_y\in\caB_y$. Then Eq. \eqref{comm} holds with $c=1$ (see Theorem \ref{thm general}).

For the third correlator, we choose $S_x :=\frac{1}{2}( n_{\uparrow, x}+ n_{\downarrow, x}$) as before.
Let $O_{xy} = c^\dagger_{\sigma, x}\mathcal{O}_y$, for any choice of $\sigma$, and an arbitrary 
$\mathcal{O}_y\in\caB_y$. Then Eq. \eqref{comm} holds with $c=2$ (see Theorem \ref{thm general}).
 
The result then follows from straightforward application of Theorem \ref{thm general} to all three cases. Indeed, let $\xi_K(\beta)$ be as in Eq.\ \eqref{xi_K}. Given the bound in Eq.\ \eqref{boundnormHtJM} and the values of $c$ in the three cases considered here, we find that
\be
\xi_K(\beta) \geq \frac{K}{2\beta}\left(1-16K\gamma^2(2|t|+|J|)2^{\frac{2K}{\beta}}\right)=\tilde{\xi}_K(\beta).
\ee
Obviously
\be\lim_{\beta\rightarrow\infty}\beta\,\tilde{\xi}_K(\beta) = \frac{K}{2}(1-16 K\gamma^2 (2|t|+|J|)).\ee
The theorem now follows by optimizing in $K$ and by defining $\xi(\beta)=\tilde{\xi}_{K^*}(\beta)$, where $K^*$ is the optimal value of $K$. We set $\mathcal{O}_y :=c^\dagger_{\downarrow,y}c_{\uparrow,y} $, in the first case, $\mathcal{O}_y :=c_{\uparrow, y}c_{\downarrow,y}$, in the second case, and $\mathcal{O}_y = c_{\sigma,y}$, in the last case.
\end{proof}

\paragraph{Acknowledgment}
We are grateful to the referee for useful comments.

\end{document}